\documentclass[12pt]{amsart}
\usepackage{amsopn}
\usepackage{amssymb, amscd}
\usepackage{mathrsfs}

\usepackage{amsmath,amssymb}

\usepackage[english]{babel}

\usepackage{graphics,color}

\usepackage{amsthm}
\usepackage{amsmath}

\usepackage{graphicx}

\newcommand{\nc}{\newcommand}

\nc{\SO}{\mathrm{SO}} \nc{\Spe}{\mathrm{Sp}} \nc{\Sl}{\mathrm{SL}}
\nc{\SU}{\mathrm{SU}} \nc{\Or}{\mathrm{O}} \nc{\U}{\mathrm{U}}
\nc{\Gl}{\mathrm{GL}} \nc{\Se}{\mathrm{S}} \nc{\Cl}{\mathrm{Cl}}
\nc{\Spein}{\mathrm{Spin}} \nc{\Pin}{\mathrm{Pin}}
\nc{\G}{\mathrm{GL}_n(\RR)} \nc{\g}{\mathfrak{gl}_n(\RR)}

\nc{\vs}{\vspace{.2cm}} \nc{\vsp}{\vspace{1cm}}

\nc{\sqa}{\mathscr{S}^{a}_q}
\nc{\sq}{\mathscr{S}_q}
\nc{\sqao}{\mathscr{S}^{a\mathscr{O}}_q}
\nc{\sqo}{\mathscr{S}^{\mathscr{O}}_q}
\nc{\sqan}{\mathscr{S}^{a}_{q,N}}
\nc{\sqn}{\mathscr{S}_{q,N}}

\numberwithin{equation}{section}

\theoremstyle{plain}
\newtheorem{theorem}{Theorem}[section]
\newtheorem{proposition}[theorem]{Proposition}
\newtheorem{corollary}[theorem]{Corollary}
\newtheorem{lemma}[theorem]{Lemma}

\theoremstyle{definition}

\theoremstyle{remark}
\newtheorem{remark}[theorem]{Remark}

\numberwithin{theorem}{section}

\title{Lie subalgebras of the matrix quantum pseudo differential operators }

\author{Karina Batistelli - Carina Boyallian}

\address{FaMAF y CIEM, Universidad Nacional de C\'ordoba, C\'ordoba, Argentina}
\email{khbatistelli@gmail.com - boyallia@mate.uncor.edu}

\begin{document}


\maketitle

\begin{abstract}
We give a complete description of the anti-involutions that preserve the  principal gradation of the algebra of matrix quantum pseudodifferential operators  and we describe the Lie subalgebras of its minus fixed points.
\end{abstract}


\section{Introduction}

The $W$-infinity algebras naturally arise in various physical theories, such as conformal field theory, the theory of quantum Hall effect, etc.  The $W_{1+\infty}$ algebra, which is the central extension of the Lie algebra $D$ of differential operators on the circle, is the most fundamental among these algebras.
The representations of the Lie algebra $W_{1+\infty}$ were first studied in \cite{KR}, where its irreducible quasifinite highest weight representations were characterized. At the end of that article, similar results were found for the central extension of the Lie algebra of quantum pseudo-differential operatos $\sq$, which contains as a subalgebra the $q-$analogue of the Lie algebra $\hat{D}$ , the algebra of all regular difference operators on $\mathbb{C}^{\times}$. Here and further $q$ is not a root of unity.

In \cite{KWY}, certain subalgebras of the Lie algebra $D$ were considered, and it was shown that there are, up to conjugation, two anti-involutions $\sigma_{\pm}$ on $D$, which preserve the principal gradation. These results were extended to the matrix case in \cite{BL01}, where a complete description of the anti-involutions of the algebra $D^N$ of the $N \times N$-matrix differential operators on the circle preserving the principal $\mathbb{Z}$-gradation was given.

Analogously, in \cite{BL05} it was shown that there is a family of anti-involutions $\sigma_{\epsilon, k}$ on $\sq$ ($\epsilon =\pm 1$, $k \, \in \, \mathbb{Z}$), up to conjugation, preserving the principal gradation. The goal of this paper is to extend these results to the matrix case, where the global image seems to be richer and more complex.

The paper is organized as follows: In Sect. 2 we give a complete description of the anti-involutions of the algebra $\sqn$ of $N \times N$-quantum pseudodifferential operators, preserving the principal $\mathbb{Z}$-gradation. For each $n \, \in \, \mathbb{N}$ with $n \leq N$ we obtain, up to conjugation, two families of anti-involutions that show quite different results when $n=N$ and $n<N$. To exhibit their differences in detail, they are studied separately in Sect. 3 and 4 respectively.  In Sect. 3 the anti-involutions give us two families $\sqn^{\epsilon, r, N}$ of Lie subalgebras ($\epsilon =\pm 1$, $r, \, k \, \in \, \mathbb{Z}$) fixed by $-\sigma_{\epsilon, r,N}$. Then we give a geometric realization of $\sigma_{\epsilon, r,N}$, concluding that $\sqn^{+,r,N}$ is a subalgebra of $\sqn$ of orthogonal type and $\sqn^{-,r,N}$ is a subalgebra of $\sqn$ of symplectic type. In Sect. 4, the families $\sigma_{\epsilon, n}$,  with $1 \leq n < N$  give us two families of Lie subalgebras $\sqn^{\epsilon, n}$ ($\epsilon =\pm 1$) fixed by $-\sigma_{\epsilon, n}$. We give a geometric realization of $\sigma_{\epsilon, n}$, concluding that $\sqn^{+,n}$ is a subalgebra of $\sqn$ of type $o(n,N-n)$ and $\sqn^{-,n}$ is a subalgebra of $\sqn$ of type $osp(n,N-n)$.

\section{Quantum Pseudo-Differential Operators}

Consider $\mathbb{C}[z,z^{-1}]$ the Laurent polynomial algebra in one variable. 
We denote $\sqa$ the associative algebra of quantum pseudo-differential operators.
Explicitly, let $T_q$ denote the  operator on $\mathbb{C}[z,z^{-1}]$ given by

\begin{equation*}
T_q f(z)=f(qz),
\end{equation*}
where $q \, \in \, \mathbb{C}^{\times}=\mathbb{C} \backslash \{0\}$. An element of $\sqa$ can be written as a linear combination of operators of the form $z^{k}f(T_q)$, where $f$ is a Laurent polynomial in $T_q$. The product in $\sqa$ is given by
\begin{equation}
(z^{m}f(T_q))(z^{k}g(T_q))=z^{m+k}f(q^{k}T_q)g(T_q).
\end{equation}

Now let $\sq$ denote the Lie algebra obtained from $\sqa$ by taking the usual commutator. Let  $\mathscr{S}'_q=[\sq,\sq]$. It follows:

\begin{equation*}
\sq=\mathscr{S}'_q \bigoplus \mathbb{C}T^{0}_q \quad \hbox{ (direct sum of ideals).}
\end{equation*}



Let $N$ be a positive integer. As of this point, we shall denote by $Mat_N K$ the associative algebra of all $N \times N$ matrices over an algebra $K$ and $E_{ij}$ the standard basis of $Mat_N \mathbb{C}$.

Let $\sqan=\sqa \otimes Mat_N \mathbb{C}$ be the associative algebra of all quantum matrix pseudodifferential operators, namely the operators on $\mathbb{C}^N[z, z^{-1}]$ of the form

\begin{equation}
E=e_k(z)T^k_q+e_{k-1}(z)T^{k-1}_q + \cdots + e_0(z),
\hbox{ where } e_k(z) \, \in \, Mat_N \mathbb{C}[z, z^{-1}].
\end{equation}

In a more useful notation, we write the pseudodifferential operators as linear combinations of elements of the form $z^k f(T_q)A$, where $f$ is a Laurent polynomial, $k \, \in \, \mathbb{Z}$ and $A \, \in \, Mat_N \mathbb{C} $. The product in $\sqan$ is given by

\begin{equation}
(z^{m}f(T_q)A)(z^{k}g(T_q)B)=z^{m+k}f(q^{k}T_q)g(T_q)AB.
\end{equation}

Let $\sqn$ denote the Lie algebra obtained from $\sqan$ with the bracket given by the commutator, namely:

\begin{equation}
[z^{m}f(T_q)A, z^{k}g(T_q)B]=z^{m+k}(f(q^{k}T_q)g(T_q)AB-f(T_q)g(q^m T_q)BA).
\end{equation}

The elements $z^k T^m_q E_{ij} \, (k \, \in \, \mathbb{Z}, m \, \in \, \mathbb{Z_{+}}, i,j \, \in \, \{1,\cdots,N\})$ form a basis of $\sqn$.

%
%
%
%
%

Define the \textit{weight} on $\sqn$ by

\begin{equation}
wt z^k f(T_q)E_{ij}=kN+i-j.
\end{equation}

This gives  the \textit{principal} $\mathbb{Z}$ -gradation of $\sqan$ and $\sqn$, the latter of which is given by $\sqn = \bigoplus_{j \, \in \, \mathbb{Z}} \mathscr{S}_{q,N,j}$. This allows the following triangular decomposition

\begin{equation*}
\sqn= \mathscr{S}_{q,N,+} \bigoplus \mathscr{S}_{q,N,0} \bigoplus \mathscr{S}_{q,N,-},
\end{equation*}

where $\mathscr{S}_{q,N,+}= \bigoplus_{j \, \in \, \mathbb{ Z}_{<0}} \mathscr{S}_{q,N,j}$ and $\mathscr{S}_{q,N,-}= \bigoplus_{j \, \in \, \mathbb{ Z}_{<0}} \mathscr{S}_{q,N,j}$.

An $\textit{anti-involution}$ $\sigma$ of $\sqan$ is an involutive anti-automorphism of $\sqan$, ie,
$\sigma^2= Id$, $\sigma(\mathit{a}x+\mathit{b}y)=\mathit(a)\sigma(x)+\mathit(b)\sigma(y)$ and $\sigma(xy)=\sigma(y)\sigma(x)$, for all $a, \, b \, \in \, \mathbb{C}$ and $x, y \, \in \, \sqan$. From now on we will assume that $|q|\neq1$.

As we intend to classify the anti-involutions of $\sqan$ preserving its principal gradation, we shall introduce some notation. For each $n$, $ 1 \leq n \leq N$, define the permutation $\pi_n$ in $S_N$ by

\begin{equation}\label{eq:matriz}
\begin{array}{ccccccccc}
1 & 2 & \cdots & n-1 & n & n+1 & \cdots & N-1 & N \\
\downarrow & \downarrow &  & \downarrow & \downarrow & \downarrow & & \downarrow & \downarrow \\
n & n-1 & \cdots & 2 & 1 & N & \cdots & n+2 & n+1
\end{array}
\end{equation}

Let us fix $n$, $1 \leq n < N$,   $B$  and $c=\{c_{i,j}\}$, $c_{i,j} \, \in \, \mathbb{C}$, $i>j$, and write

\begin{equation*}
\delta_{i \leq n} = \begin{cases} 1 & \text{ if } i \leq n \\
								 0 & \text{ if } i > n.
\end{cases}
\end{equation*}

We define $\sigma=\sigma_{\pm,B,c,n}$ in $\sqan$ by the following formulas:

\begin{equation}\label{eq:generators}
\sigma(E_{ii})=E_{\pi_n(i),\pi_n(i)},
\end{equation}

\begin{equation*}
\sigma(zE_{ii})=  \pm zE_{\pi_n(i),\pi_n(i)}
\end{equation*}


\begin{equation*}
\sigma(T_qE_{ii})=Bq^{-1+\delta_{i \leq n}}T^{-1}_qE_{\pi_n(i),\pi_n(i)}
\end{equation*}

\vspace{5mm} 
\begin{equation*}
(i>j) \quad \sigma(E_{i,j})= \begin{cases}
c_{i,j} E_{\pi_n(j),\pi_n(i)} & \text{if} \quad i \leq n \quad \textit{or} \quad  j>n  \\
zc_{i,j} E_{\pi_n(j),\pi_n(i)} & \text{if} \quad i>n \quad \textit{and} \quad  j \leq n
\end{cases},
\end{equation*}

\vspace{5mm} 

\begin{equation*}
(i<j) \quad \sigma(E_{i,j})= \begin{cases}
c^{-1}_{j,i} E_{\pi_n(j),\pi_n(i)} & \text{if} \quad i>n \quad \textit{or} \quad  j \leq n  \\
z^{-1}c^{-1}_{j,i} E_{\pi_n(j),\pi_n(i)} & \text{if} \quad i \leq n \quad \textit{and} \quad  j>n
\end{cases},
\end{equation*}

\vspace{5mm} 


\begin{theorem} \label{Teoremita}
Let $1 \leq n < N$. $\sigma=\sigma_{\pm,B,c,n}$ defined on generators by ~\eqref{eq:generators} extends to an anti-involution on $\sqan$ which preserves the principal $\mathbb{Z}$-gradation if and only if

\begin{subequations}
\begin{equation} \label{eq:productosc}
c_{ij}=c_{i,i-1}c_{i-1,i-2}\cdots c_{j+1,j}
\end{equation}

\begin{equation}\label{eq:lamolesta}
\begin{cases} c_{i,j}c_{\pi_n(j),\pi_n(i)}=1 & \text{if} \quad i \leq n \quad \textit{or} \quad j>n \\
			  c_{i,j}c^{-1}_{\pi_n(i),\pi_n(j)}=\pm 1 & \text{if} \quad i>n \quad \textit{and} \quad  j \leq n
\end{cases}
\end{equation}
\end{subequations}

Moreover, any anti-involution $\sigma$ of $\sqan$ which preserves the principal $\mathbb{Z}$-gradation is one of them.
\end{theorem}



The proof will mainly consist of several steps making use of the involutive property of $\sigma$ and the relations between the generators $E_{i,j}$.

$\newline$
\begin{proof}
Fix $i \, \in \, \{ 1, \cdots, N \}$.


$\mathit{Step \, 1.}$

Because $\sigma$ should preserve the principal $\mathbb{Z}$-gradation, we have $\sigma(E_{i,i})=\sum^N_{j=1} U_{i,j}(T_q) E_{j,j}$. Given the fact that $\sigma$ is an anti-involution, we get $\sigma(E_{i,i})=\sigma(E_{i,i}E_{i,i})=\sigma(E_{i,i})\sigma(E_{i,i})=\sum^N_{j=1} (U_{i,j}(T_q))^2E_{j,j}$, so $U_{i,j}(T_q)=U^2_{i,j}(T_q)$. Taking into consideration the positive and negative degrees of these Laurent polynomials, we arrive at $U_{i,j}(T_q)=a_{i,j}$, where $a_{i,j}$ are constant elements such that $a^2_{i,j}=a_{i,j}$. This gives us $a_{i,j}=0$ or $a_{i,j}=1$ for every $i,j \, \in \, \{1,...,N\}$.
We also know that $E_{i,i}=\sigma^2(E_{i,i})=\sum^N_{j=1}a_{i,j}\sum^N_{k=1}a_{j,k}E_{k,k}$. So,  $1=\sum^N_{j=1}a_{i,j}a_{j,i}$ and $0=\sum^N_{j=1}a_{i,j}a_{j,k}$ for $k \neq i$. So, for each $i$ there exists a unique $j_i$ such that $a_{i,j_i}=a_{j_i,i}=1$ and $a_{i,j}a_{j,i}=0$ for any $j \neq j_i$. And $a_{i,j}a_{j,k}=0$ for every $j, i$ and $k \neq i$. In particular, $a_{i,j_i}a_{j_i,k}=0$ for $k \neq i$, so $a_{j_i,k}=0$ for any $k \neq i$, obtaining that $\sigma(E_{i,i})=E_{j_{i},j_{i}}$. Due to the injectivity of $\sigma$, $\pi(i):= j_{i}$ is a permutation in $S_N$, and since $\sigma$ is an involution, we have $\pi^2=id$.
 $ $\newline

 $\mathit{Step \, 2}$.

 Again, due to the fact that $\sigma$ should preserve the principal $\mathbb{Z}$-gradation, we may assume that
$\sigma(T_qE_{i,i})=\sum^N_{j=1} P_{i,j}(T_q) E_{j,j}$ and $\sigma(T_q^{-1}E_{i,i})=\sum^N_{j=1} H_{i,j}(T_q) E_{j,j}$. So,

\begin{align*}
\sigma(T_qE_{i,i})
&= \sigma(T_qE_{i,i}E_{i,i})=\sigma(E_{i,i})\sigma(T_qE_{i,i}) \\
&= E_{\pi(i), \pi(i)} \left( \sum^N_{j=1} P_{i,j}(T_q)E_{j,j} \right)	\\
&= P_{i,\pi(i)}(T_q)E_{\pi(i), \pi(i)}.
\end{align*}

Proceding similarly with  $\sigma(T_q^{-1}E_{i,i})$, we have
\begin{equation*}
\sigma(T_q^{-1}E_{i,i}) =H_{i,\pi(i)}(T_q)E_{\pi(i), \pi(i)}.
\end{equation*}

Combining these two equations, we have
\begin{align*}
E_{\pi(i), \pi(i)}
&=\sigma(E_{i,i})=\sigma(T_q^{-1}E_{i,i}T_qE_{i,i}) \\
&=\sigma(T_qE_{i,i})\sigma(T_q^{-1}E_{i,i}) \\
&=P_{i,\pi(i)}(T_q)H_{i, \pi(i)}(T_q) E_{\pi(i), \pi(i)}
\end{align*}

So, $1=P_{i,\pi(i)}(T_q)H_{i,\pi(i)}(T_q)$ and, as consecuence, they must be units of the Laurent polynomial ring. Therefore, we can assume $P_{i}(T_q):= P_{i,\pi(i)}(T_q)=B_iT^{k_i}_q$ and $H_{i,\pi(i)}(T_q)=B^{-1}_iT^{-{k_i}}_q$, with $B_i \, \in \, \mathbb{C^\times}$ and $k_i \, \in \, \mathbb{Z}$.
So, $\sigma(T_q^{-1}E_{i,i})$ is then determined by $\sigma(T_qE_{i,i})$.

Now, let us note that we can write
$T^k_qE_{i,i}=T^{k-1}_qE_{i,i}T_qE_{i,i}= \cdots =(T_q E_{i,i})^k$, for every $k \, \in \, \mathbb{Z}$.
Therefore,

\begin{align}
T_qE_{i,i}
&=\sigma^2(T_qE_{i,i})=\sigma(B_i T^{k_i}_q E_{\pi(i),\pi(i)}) \nonumber \\
&=B_i(\sigma(T_qE_{\pi(i), \pi(i)}))^{k_i} =B_iB^{k_i}_{\pi(i)}T^{k_ik_{\pi(i)}}_qE_{i,i}
\end{align}

So, $B_iB^{k_i}_{\pi(i)}=1$ and $k_ik_{\pi(i)}=1$ This gives us the following alternatives $k_i=k_{\pi(i)}=1$ or $k_i=k_{\pi(i)}=-1$.

$ $\newline

$\mathit{Step \, 3}$.

Since $wt(z^{\pm 1}E_{i,i})=\pm N$ and $\sigma$ should preserve the principal $\mathbb{Z}$-gradation, we can assume
$\sigma(zE_{i,i})=z\sum^N_{j=1}T_{i,j}(T_q)E_{j,j}$ and $\sigma(z^{-1}E_{i,i})=z^{-1}\sum^N_{j=1}\hat{T}_{i,j}(T_q)E_{j,j}$. Using a similar argument to the one used in Step 2 and denoting $T_{i,\pi(i)}(T_q) := T_{i}(T_q)$ and $\hat{T}_{i,\pi(i)}(T_q) := \hat{T}_{i}(T_q)$, we can deduce that $\hat{T}_j$ and $T_j=0$ for $j \neq \pi(i)$, and also $T_i(T_q)=A_i T^{r_i}_q$ and $\hat{T}_{i}(T_q)=C_iT^{-r_i}_q$, with $C_i=A^{-1}_iq^{r_i}$, $A_i \, \in \, \mathbb{C^{\times}}$ and $r_i \, \in \, \mathbb{Z}$. So,

\begin{align*}
zE_{i,i}
&=\sigma^2({zE_{i,i}})=\sigma(zA_iT^{r_i}_qE_{\pi(i), \pi(i)}) \\
&=A_i\sigma(T_qE_{\pi(i), \pi(i)})^{r_i}\sigma(zE_{\pi(i), \pi(i)}) \\
&=A_i A_{\pi(i)}B^{r_i}_{\pi(i)}q^{k_{\pi(i)}r_i}zT^{r_ik_{\pi(i)}+r_{\pi(i)}}_qE_{i,i}
\end{align*}

Therefore, we have $r_ik_{\pi(i)}+r_{\pi(i)}=0$ and $A_i A_{\pi(i)}B^{r_i}_{\pi(i)}q^{k_{\pi(i)}r_i}=1$. On the other hand,

\begin{align*}
zT_qE_{i,i}
&=\sigma^2({zT_qE_{i,i}})=\sigma(zA_iB_iT^{r_i+k_i}_qq^{k_i}E_{\pi(i), \pi(i)}) \\
&=A_iB_iq^{k_i}\sigma(T_qE_{\pi(i), \pi(i)})^{k_i+r_i}\sigma(zE_{\pi(i), \pi(i)}) \\
&=A_i A_{\pi(i)}B_i B^{r_i+k_i}_{\pi(i)}q^{k_i+k_{\pi(i)}(r_i+k_i)}zT^{k_{\pi(i)}(k_i+r_i)+r_{\pi(i)}}_qE_{i,i}
\end{align*}

We can therefore conclude that $k_{\pi(i)}(k_i+r_i)+r_{\pi(i)}=1$ and $A_i A_{\pi(i)}B_i B^{r_i+k_i}_{\pi(i)}q^{k_i+k_{\pi(i)}(r_i+k_i)}=1$. From this last equation and the previous step, we get $1=q^{k_i+1}$.

If $k_i=1$, then $q^2=1$. Since we assumed that $q$ is not a root of unity, it is easy to check that these are not anti-automorphisms. Therefore, $k_i=-1$ , $B_i = B_{\pi(i)}$ and $r_{\pi(i)}=r_i$.

By now,
\begin{equation*}
\sigma(E_{ii})=E_{\pi(i),\pi(i)},
\end{equation*}

\begin{equation*}
\sigma(zE_{ii})=A_izT^{r_i}_qE_{\pi(i),\pi(i)},
\end{equation*}


\begin{equation*}
\sigma(z^{-1}E_{ii})=A^{-1}_iq^{r_i}z^{-1}T^{r_i}_qE_{\pi(i),\pi(i)},
\end{equation*}


\begin{equation*}
\sigma(T_qE_{ii})=B_iT^{-1}_qE_{\pi(i),\pi(i)}
\end{equation*}

where
\begin{equation}\label{eq:aiapi}
A_i A_{\pi(i)}B^{r_i}_{j}q^{-r_i}=1,
\end{equation}

and also $B_i = B_{\pi(i)}$ and $r_{\pi(i)}=r_i$, for $A_i, A_{\pi_n(i)},B_i \, \in \, \mathbb{C^{\times}}$ and $r_i \, \in \, \mathbb{Z}$.

$ $\newline

$\mathit{Step 4}$.
Suppose $i>j$. As an implication of the $\mathbb{Z}$-gradation preservation property of $\sigma$, we have that

$\sigma(E_{i,j})=\sum^{N-i+j}_{l=1}C^{i,j}_l(T_q)E_{l+i-j,l}+\sum^{N}_{l=N-i+j+1}z\hat{C}^{i,j}_l(T_q)E_{l+i-j-N,l}$. Since

$\sigma(E_{i,i}E_{i,j})=\sigma(E_{i,j})\sigma(E_{i,i})=\sigma(E_{i,j})E_{\pi(i), \pi(i)}$, we can deduce

\begin{equation} \label{eq:(2.9)}
\sigma(E_{i,j})= \begin{cases} C^{i,j}(T_q)E_{\pi(i)+i-j,\pi(i)} & \text{if} \quad   \pi(i) \leq N-i+j \\
								z\hat{C}^{i,j}(T_q)E_{\pi(i)+i-j-N,\pi(i)} & \text{if} \quad   \pi(i) \geq N-i+j+1
								\end{cases}								
\end{equation}

where $C^{i,j}(T_q)=C^{i,j}_{\pi(i)}(T_q)$, and $\hat{C}^{i,j}(T_q)=\hat{C}^{i,j}_{\pi(i)}(T_q)$.

Similarly, if $i<j$ and
$\sigma(E_{i,j})=\sum^{j-i}_{l=1}z^{-1}\hat{S}^{i,j}_l(T_q)E_{N+l+i-j,l}+\sum^{N}_{l=j-i+1}S^{i,j}_l(T_q)E_{l+i-j,l}$,
we deduce

\begin{equation} \label{eq:(2.10)}
\sigma(E_{i,j})= \begin{cases} S^{i,j}(T_q)E_{\pi(i)+i-j,\pi(i)} & \text{if} \quad \pi(i) \geq j-i+1 \\
								z^{-1}\hat{S}^{i,j}(T_q)E_{\pi(i)+i-j+N,\pi(i)} & \text{if} \quad \pi(i) \leq j-i
							\end{cases}								
\end{equation}

where $S^{i,j}(T_q)=S^{i,j}_{\pi(i)}(T_q)$.

$\mathit{Case \, a}$. Let $i>j$, with $\pi(i) \leq N-i+j$:

\begin{equation*}
E_{i,j}=\sigma^2(E_{i,j})=\sigma(C^{i,j}(T_q)E_{\pi(i)+i-j,\pi(i)})=\sigma(C^{i,j}(T_q)E_{\pi(i)+i-j,\pi(i)+i-j}E_{\pi(i)+i-j,\pi(i)})
\end{equation*}

using \eqref{eq:(2.9)}, we must have $\pi(\pi(i)+i-j) \leq N-i+j $ because we would otherwise get $z$ in the right-hand side above, so

\begin{align*}
E_{i,j}
&=\sigma(E_{\pi(i)+i-j,\pi(i)})\sigma(C^{i,j}(T_q)E_{\pi(i)+i-j,\pi(i)+i-j}) \\
&=(C^{\pi(i)+i-j, \pi(i)}(T_q)E_{\pi(\pi(i)+i-j)+i-j, \pi(\pi(i)+i-j)} )(C^{i,j}(B_{\pi(i)+i-j}T^{-1}_q)E_{\pi(\pi(i)+i-j), \pi(\pi(i)+i-j)}) \\
&=C^{\pi(i)+i-j, \pi(i)}(T_q) C^{i,j}(B_{\pi(i)+i-j}T^{-1}_q)E_{\pi(\pi(i)+i-j)+i-j, \pi(\pi(i)+i-j)}.
\end{align*}

Then, $C^{i,j}(B_{\pi(i)+i-j}T^{-1}_q)$ and $C^{\pi(i)+i-j, \pi(i)}(T_q)$ are units of the Laurent polynomial ring and $\pi(j)=\pi(i)+i-j$ . Therefore, and because $B_i= B_{\pi(i)}$, we can write $C^{\pi(j), \pi(i)}(T_q) = c_{\pi(j), \pi(i)}T^{s_{\pi(j), \pi(i)}}_q$ and $C^{i,j}(B_jT^{-1}_q) = c_{i,j}T^{s_{i,j}}_q$, with

\begin{equation}\label{eq:latranca}
c_{i,j}.c_{\pi(j), \pi(i)}=1 \quad \text{ and }  -s_{\pi(j), \pi(i)}=s_{i,j}, \quad \text{ thus } \quad C^{i,j}(T_q) = c_{i,j}B^{s_{i,j}}_j T^{-s_{i,j}}_q.
\end{equation}

$ $\newline
$\mathit{Case \, b}$. Let $i>j$ and if  $\pi(i) \geq N-i+j+1$, in the same way, using simultaneously \eqref{eq:(2.9)} and \eqref{eq:(2.10)} in order to take care of the $z$ that appears in $\sigma(E_{i,j})$, we have $\pi(\pi(i)+i-j-N) \leq N + j-i$, thus:

\begin{align*}
E_{i,j}
&=\sigma^2(E_{i,j})=\sigma(zC^{i,j}(T_q)E_{\pi(i)+i-j-N,\pi(i)}) \\
&=\sigma(E_{\pi(i)+i-j-N,\pi(i)})\sigma(C^{i,j}(T_q)E_{\pi(i)+i-j-N,\pi(i)+i-j-N})\sigma(zE_{\pi(i)+i-j-N,\pi(i)+i-j-N}) \\
&=\left(z^{-1}\hat{S}^{\pi(i)+i-j-N, \pi(i)}(T_q)\right)\left(C^{i,j}(B_{\pi(i)+i-j-N}T^{-1}_q)\right) \times \\
& \qquad \qquad \qquad \qquad \qquad \left(A_{\pi(i)+i-j-N}zT^{r_{\pi(i)+i-j-N}}_q\right)
E_{\pi(\pi(i)+i-j-N)+i-j, \pi(\pi(i)+i-j-N)} \\
&=A_{\pi(i)+i-j-N}\hat{S}^{\pi(i)+i-j-N,\pi(i)}(qT_q)C^{i,j}(q^{-1}B_{\pi(i)+i-j-N}T^{-1}_q) \times \\
& \qquad \qquad \qquad \qquad \qquad \qquad \qquad \qquad \qquad \qquad \left(T^{r_i}_qE_{\pi(\pi(i)+i-j-N)+i-j, \pi(\pi(i)+i-j-N)}\right).
\end{align*}

Therefore, $j=\pi(\pi(i)+i-j-N)$ and we can assume $A_{\pi(j)}\hat{S}^{\pi(j), \pi(i)}(qT_q)T^{r_i}_q=d_{\pi(j), \pi(i)}T^{p_{\pi(j), \pi(i)}}_q$ and $C^{i,j}(q^{-1}B_jT^{-1}_q)=c_{i,j}T^{m_{i,j}}_q$, with

\begin{equation}\label{eq:n}
d_{\pi(j), \pi(i)}.c_{i,j}=1 \quad  \text{ and } \quad p_{\pi(j), \pi(i)}=-m_{i,j}, \quad \text{ thus } \quad C^{i,j}(T_q) = c_{i,j}q^{-m_{i,j}}B^{m_{i,j}}_jT^{-m_{i,j}}_q.
\end{equation}

$\mathit{Case \, c}$. Let $i<j$ and $\pi(i) \geq j-i+1$:

\begin{equation*}
E_{i,j}=\sigma^2(E_{i,j})=\sigma(S^{i,j}(T_q)E_{\pi(i)+i-j,\pi(i)})=\sigma(S^{i,j}(T_q)E_{\pi(i)+i-j,\pi(i)+i-j}E_{\pi(i)+i-j,\pi(i)})
\end{equation*}

using \eqref{eq:(2.10)}, we must have $\pi(\pi(i)+i-j) \geq j-i+1 $ in order to avoid getting $z^{-1}$ in the right-hand side above. So

\begin{align*}
E_{i,j}
&=\sigma(E_{\pi(i)+i-j,\pi(i)})\sigma(S^{i,j}(T_q)E_{\pi(i)+i-j,\pi(i)+i-j}) \\
&=\left(S^{\pi(i)+i-j, \pi(i)}(T_q)E_{\pi(\pi(i)+i-j)+i-j, \pi(\pi(i)+i-j)} \right)\left(S^{i,j}(B_{\pi(i)+i-j}T^{-1}_q)E_{\pi(\pi(i)+i-j), \pi(\pi(i)+i-j)}\right) \\
&=S^{\pi(i)+i-j, \pi(i)}(T_q) S^{i,j}(B_{\pi(i)+i-j}T^{-1}_q)E_{\pi(\pi(i)+i-j)+i-j, \pi(\pi(i)+i-j)}.
\end{align*}

Then, $j=\pi(\pi(i)+i-j)$ and $S^{i,j}(B_{j}T^{-1}_q)$ and $S^{\pi(i)+i-j, \pi(i)}(T_q)$ are units of the Laurent polynomial ring, so we can assume $S^{\pi(j),\pi(i)}(T_q)=d_{\pi(j),\pi(i)}T^{u_{\pi(j),\pi(i)}}_q$ and $S^{i,j}(B_{j}T^{-1}_q)=d_{i,j}T^{u_{i,j}}_q$, with

\begin{equation*}
d_{i,j}.d_{\pi(j), \pi(i)}=1 \quad \text{ and }  \quad u_{\pi(j), \pi(i)}=-u_{i,j},  \quad \text{ thus } \quad  S^{i,j}(T_q)=d_{i,j}B^{u_{i,j}}_jT^{-u_{i,j}}_q.
\end{equation*}

$ $\newline

$\mathit{Case \, d}$. Let $i<j$ and if  $\pi(i) \leq j-i$, since $\sigma$ is an involution, we make use of \eqref{eq:(2.9)} and \eqref{eq:(2.10)} simultaneously to take care of the $z^{-1}$ appearing in $\sigma(E_{i,j})$. In order to do this, we require $\pi(N+\pi(i)+i-j) \geq -i+j+1$. So:

\begin{align*}
E_{i,j}
&=\sigma^2(E_{i,j})=\sigma(z^{-1}\hat{S}^{i,j}(T_q)E_{\pi(i)+i-j+N,\pi(i)}) \\
&=\sigma(E_{\pi(i)+i-j+N,\pi(i)})\sigma(\hat{S}^{i,j}(T_q)E_{\pi(i)+i-j+N,\pi(i)+i-j+N})\sigma(z^{-1}E_{\pi(i)+i-j+N,\pi(i)+i-j+N}) \\
&=\left(z\hat{C}^{\pi(i)+i-j+N, \pi(i)}(T_q)\right)\left(\hat{S}^{i,j}(B_{\pi(i)+i-j+N}T^{-1}_q)\right) \times \\
& \qquad \qquad \qquad \qquad \left(A^{-1}_{\pi(i)+i-j+N}q^{r_{\pi(i)+i-j+N}}z^{-1}T^{-r_{\pi(i)+i-j+N}}_q\right)E_{\pi(\pi(i)+i-j+N)+i-j, \pi(\pi(i)+i-j+N)} \\
&=A^{-1}_{\pi(i)+i-j+N}q^{r_{\pi(i)+i-j+N}}\hat{C}^{\pi(i)+i-j+N,\pi(i)}(q^{-1}T_q)\hat{S}^{i,j}(qB_{\pi(i)+i-j+N}T^{-1}_q) \times \\
&\qquad \qquad \qquad \qquad\qquad \qquad \qquad \qquad \qquad \qquad T^{-r_{\pi(i)+i-j+N}}_qE_{\pi(\pi(i)+i-j+N)+i-j, \pi(\pi(i)+i-j+N)}.
\end{align*}

Then, $j=\pi(\pi(i)+i-j+N)$. Once again, being $\hat{S}^{i,j}(qB_{j}T^{-1}_q)$ and

$A^{-1}_{\pi(j)}q^{r_{\pi(j)}}\hat{C}^{\pi(j), \pi(i)}(q^{-1}T_q)T^{-r_{\pi(j)}}_q$ units of the Laurent polynomial ring, we can write

$\hat{S}^{i,j}(qB_jT^{-1}_q)=d_{i,j}T^{b_{i,j}}_q$ and $A^{-1}_{\pi(j)}q^{r_{\pi(j)}}\hat{C}^{\pi(j), \pi(i)}(q^{-1}T_q)T^{-r_{\pi(j)}}_q=c_{\pi(j), \pi(i)}T^{\epsilon_{\pi(j), \pi(i)}}_q$, with

\begin{equation}\label{eq:otrodelosmolestos}
1= d_{i,j} c_{\pi(j), \pi(i)} \quad \text{ and } \quad -\epsilon_{\pi(j), \pi(i)}=b_{i,j}, \quad \text{ thus } \quad \hat{S}^{i,j}(T_q)=d_{i,j}q^{b_{i,j}}B^{b_{i,j}}_jT^{-b_{i,j}}_q.
\end{equation}


$\mathit{Step \, 5}$ Let $i>j$, then by Step 1, $E_{\pi(i), \pi(i)}=\sigma(E_{i,i})=\sigma(E_{i,j}E_{j,i})=\sigma(E_{j,i})\sigma(E_{i,j})$. Using \eqref{eq:(2.9)} with condition $\pi(i) \leq N-i+j$, we have  that $\pi(j)=\pi(i)+i-j$, therefore  $\pi(j) \geq i-j+1$ trivially. So,

\begin{align*}
E_{\pi(i), \pi(i)}
&=S^{j,i}(T_q)E_{\pi(j)+j-i, \pi(j)} C^{i,j}(T_q)E_{\pi(i)+i-j, \pi(i)} \\
&=d_{j,i}c_{i,j}B^{u_{j,i}}_iB^{s_{i,j}}_jT^{-u_{j,i}-s_{i,j}}_q E_{\pi(i), \pi(i)}
\end{align*}

with
\begin{equation*}
-s_{i,j}=u_{j,i} \quad \text{ and } \quad 1=d_{j,i}c_{i,j}B^{-s_{i,j}}_iB^{s_{i,j}}_j.
\end{equation*}

We can finally rewrite $S^{i,j}(T_q) = c^{-1}_{j,i}B^{-s_{j,i}}_iT^{s_{j,i}}_q$ if $i<j$ and $\pi(j) \geq i-j+1$.

$ $\newline
Now, in the case $i>j$ and $\pi(i) \geq N-i+j+1$ in \eqref{eq:(2.9)}, we have $\pi(j)=\pi(i)+i-j-N$ and it is immediate that $\pi(j) \leq i-j$, then:

\begin{align*}
E_{\pi(i), \pi(i)}
&=\sigma(E_{i,i})=\sigma(E_{j,i})\sigma(E_{i,j}) \\
&=\left(z^{-1}S^{j,i}(T_q)E_{N+\pi(j)+j-i, \pi(j)}\right) \left(zC^{i,j}(T_q)E_{\pi(i)+i-j-N, \pi(i)}\right) \\
&=\left(z^{-1}d_{j,i}q^{b_{j,i}}B^{b_{j,i}}_{i}T^{-b_{j,i}}_q E_{\pi(i), \pi(j)}\right) \left(zc_{i,j}q^{-m_{i,j}}B^{m_{i,j}}_{j}T^{-m_{i,j}}_qE_{\pi(j), \pi(i)}\right) \\
&=c_{j,i}c_{i,j}q^{-m_{i,j}}B^{b_{j,i}}_iB^{m_{i,j}}_jT^{-b_{j,i}-m_{i,j}}_q E_{\pi(i),\pi(i)}
\end{align*}

So,
\begin{equation}\label{eq:unodelosmolestos}
m_{i,j}=-b_{j,i} \quad \text{ and } \quad d_{i,j}=c^{-1}_{j,i}q^{m_{j,i}}B^{-m_{j,i}}_iB^{m_{j,i}}_j
\end{equation}

Because of this, we have $S^{i,j}(T_q) = c^{-1}_{j,i}B^{-m_{j,i}}_iT^{m_{j,i}}_q$ if $i<j$ and $\pi(j) \leq j-i$.

$ $\newline
Thus, we can rewrite \eqref{eq:(2.9)} and \eqref{eq:(2.10)} as the following, for $i>j$

\begin{equation}\label{eq:(2.11)}
\sigma(E_{i,j})= \begin{cases} c_{i,j}B^{s_{i,j}}_jT^{-s_{i,j}}_q E_{\pi(i)+i-j,\pi(i)} & \text{if} \quad   \pi(i) \leq N-i+j \\
								zc_{i,j}q^{-m_{i,j}}B^{m_{i,j}}_jT^{-m_{i,j}}_q E_{\pi(i)+i-j-N,\pi(i)} & \text{if} \quad   \pi(i) \geq N-i+j+1
								\end{cases}								
\end{equation}

and $i<j$

\begin{equation} \label{eq:(2.12)}
\sigma(E_{i,j})= \begin{cases} c^{-1}_{j,i}B^{-s_{j,i}}_iT^{s_{j,i}}_q E_{\pi(i)+i-j,\pi(i)} & \text{if} \quad \pi(i) \geq j-i+1 \\
								z^{-1}c^{-1}_{j,i}B^{-m_{j,i}}_iT^{m_{j,i}}_qE_{\pi(i)+i-j+N,\pi(i)} & \text{if} \quad \pi(i) \leq j-i
							\end{cases}								
\end{equation}


$ $\newline
We now intend to determine the permutation $\pi$. So, let $i_0$ be such that $\pi(i_0)=N$. In case a  and  case c,  $\pi(j)=\pi(i)+i-j$ and it is easy to see that

\begin{equation}\label{eq:pichi}
\pi(i-1)=\pi(i)+1 \quad \mathrm{for \quad any} \quad i \neq i_0.
\end{equation}	

Moreover, since in case b is $\pi(j)=\pi(i)+i-j-N$, we have $\pi(i_0-1)=1$. Since $\pi$ is a bijective map, we conclude that $\pi$ must be $\pi_n$ given in \eqref{eq:matriz} where $n=i_0-1$.

Let us note that if $i \leq n$, $\pi(i)=n-i+1$, and if $i>n$, $\pi(i)=N+n+1-i$. As a consequence, we can easily see that, if $i>j$, $\pi(i) \leq N-i+j$ (case a) corresponds to the choice $i \leq n$ or $j>n$, and the case in which $j<n$ and $i>n$ corresponds to $\pi(i)>N-i+j$ (case b). Similarly for $i<j$, when $i>n$ or $j \leq n$, we have $\pi(i)>j-i$ (case c) and the case $i \leq n$ and $j>n$ corresponds to $\pi(i) \leq j-i$ (case d).

$ $\newline
Computing $z^kT^l_qE_{i,j}=\sigma^2(z^kT^l_qE_{i,j})$ in the four cases for $k \geq 0$ and $l \geq 0$, with their corresponding restrictions, we have:

In case a, where $i>j$ and $i<n$ or $j>n$, we get:

\begin{equation} \label{eq:casoa}
c_{i,j}c_{\pi(j),\pi(i)}B^l_iB^{-l+kr_i}_jA^k_iA^k_{\pi(j)}q^{-kr_i}=1
\end{equation}

and

\begin{equation} \label{eq:casoaa}
-s_{\pi(i),\pi(j)}+s_{i,j}-kr_i+kr_{\pi(j)}=0
\end{equation}

Regarding \eqref{eq:casoaa} when $k = 0$ we deduce, combining \eqref{eq:casoaa} with \eqref{eq:latranca}, that $s_{i,j}=0$.

On the other hand, $k = 1$ in \eqref{eq:casoaa} combined with the fact that $r_{\pi(j)}=r_j$, we get $r_i = r_j$.
So, $r_i = \begin{cases} r & i \leq n \\ \tilde{r} & i>n \end{cases}$.

Now, due to \eqref{eq:latranca} and \eqref{eq:casoa}, with $k = 0$, $l = 1$: $B_i = B_j$.
So, $B_i = \begin{cases} B & i \leq n \\ \tilde{B} & i>n \end{cases}$.

$ $\newline
If we consider $k = 1$, $l = 0$ in \eqref{eq:casoa} and \eqref{eq:latranca}, we have $B^{r_j}_iq^{-r_i}A_iA_{\pi(j)}=1$. Using \eqref{eq:aiapi} and the fact that $A_{\pi(j)}=A_j$, we get $A_i = A_{j}$. So, $A_i = \begin{cases} A & i \leq n \\ \tilde{A} & i>n \end{cases}$. Thus,

\begin{equation} \label{eq:ladea}
A^2(Bq^{-1})^r = 1
\end{equation}

resembling \cite{BL05}, and

\begin{equation} \label{eq:ladeatilde}
\tilde{A}^2(\tilde{B}q^{-1})^{\tilde{r}} = 1.
\end{equation}

$ $\newline
In case b, where $i>j$ and $i>n>j$, we have:

\begin{equation} \label{eq:casob}
c_{i,j}c^{-1}_{\pi(i), \pi(j)}B^{-l+k\tilde{r}-m_{\pi(i),\pi(j)}} \tilde{B}^l \tilde{A}^k A^{k+1} q^ {l-2\tilde{r}k+rk(k-1)/2+(k+1)m_{\pi(i),\pi(j)} }=1
\end{equation}

and
\begin{equation}\label{eq:casobb}
m_{\pi(i), \pi(j)}+m_{i,j}-k\tilde{r}+(k+1)r=0
\end{equation}

Regarding \eqref{eq:casobb} when $k =0$, we deduce that

\begin{equation} \label{eq:rmgral}
m_{\pi(i), \pi(j)}+m_{i,j}+r=0.
\end{equation}

$ $\newline
On the other hand, when $k = 1$ in \eqref{eq:casobb}: $m_{\pi(i), \pi(j)}+m_{i,j}-\tilde{r}+2r=0$. Combining the last two items we get $\tilde{r}=r$.

$ $\newline
Now, due to \eqref{eq:chats} and \eqref{eq:casob} with $k = 0$, $l =0$, $c_{i,j} c^{-1}_{\pi(i), \pi(j)}B^{-m_{\pi(i), \pi(j)}}Aq^{m_{\pi(i), \pi(j)}}=1.$
Combining this with \eqref{eq:casob}, we get that for arbitrary values of $k$ and $l$,

\begin{equation}\label{eq:casobreduc}
B^{-l+k\tilde{r}} \tilde{B}^l \tilde{A}^k A^{k} q^ {l-2\tilde{r}k+rk(k-1)/2+km_{\pi(i),\pi(j)} }=1
\end{equation}

$ $\newline
If we consider $k = 0$, $l =1$ in the last equation, we get $B^{-1}\tilde{B}q=1$. So, $\tilde{B}=q^{-1}B$ and due to \eqref{eq:ladea} and \eqref{eq:ladeatilde}, $\tilde{A}^2=A^2q^r$.

$ $\newline
Finally, when $k = 1$, $l =0$ in \eqref{eq:casobreduc}: $q^{m_{\pi(i), \pi(j)}}$ is constant for every $i>n>j$. So, $m_{r,s}=m$ for every $r>n>s$.

$ $\newline
Now, because of \eqref{eq:rmgral},

\begin{equation}\label{eq:myr}
2m+r=0.
\end{equation}

Letting $k=2$, $l=0$ in \eqref{eq:casobreduc}, we have $B^{2r}\tilde{A}^2A^2q^{-3r+2m}=1$. Since $\tilde{A}^2=A^2q^r$, $B^{2r}A^4q^{-2r+2m}=1$, resulting in $q^{2m}=1$ because of \eqref{eq:ladea}. So, $m=0$ and by \eqref{eq:myr}, $r=0$ and in \eqref{eq:ladea} and \eqref{eq:ladeatilde}, this implies $\tilde{A}^2=A^2=1$.

Again, letting $k=1$ and $l=0$ in \eqref{eq:casobreduc}, $\tilde{A}A=1$ and combining this with the previous equation, we get $A=\tilde{A}=\pm 1$. 

Cases c and d give the same results. 

$ $\newline
We have thus arrived at the final relations of \eqref{eq:generators}.

$ $\newline
Now, recall that we have for $1 \leq i \leq N$,
$E_{\pi(i), \pi(i)}=\sigma(E_{i,i})=\sigma(E_{i,i-1}E_{i-1,i})=\sigma(E_{i-1,i})\sigma(E_{i,i-1})$.
So, rewriting \eqref{eq:(2.11)} and \eqref{eq:(2.12)} for these cases, we have:

\begin{equation}
\sigma(E_{i,i-1})= \begin{cases} c_{i,i-1} E_{\pi(i)+1,\pi(i)} & \text{if} \quad   \pi(i) < N \\
								zc_{i,i-1}E_{1,N} & \text{if} \quad   \pi(i)=N
								\end{cases}								
\end{equation}

and

\begin{equation}
\sigma(E_{i-1,i})= \begin{cases} c^{-1}_{i,i-1}E_{\pi(i-1)-1,\pi(i-1)} & \text{if} \quad \pi(i-1) > 1 \\
								z^{-1} c^{-1}_{i,i-1}  E_{N,1} & \text{if} \quad \pi(i-1)=1
							\end{cases}								
\end{equation}

$ $\newline
If $i>j$, since $\sigma(E_{i,j}) = \sigma(E_{i,i-1}E_{i-1,i-2} \cdots E_{j+1,j})$, we get \eqref{eq:productosc}. Finally, \eqref{eq:lamolesta} are results of \eqref{eq:latranca} and of \eqref{eq:casob} with $k=0, \, l=0$ and taking into consideration that $m=0$ and $r=0$.

$ $\newline
On the other hand, it is straightforward to check that $\sigma$ defined by \eqref{eq:generators} is indeed anti-involution of $\sqa$, finishing the proof.

\end{proof}


\begin{corollary}
If $N=n$, the anti-involution $\sigma=\sigma_{A,B,c,r,N}$ is given by

\begin{equation*}
\sigma(E_{ii})=E_{\pi_n(i),\pi_n(i)}
\end{equation*}

\begin{equation*}
\sigma(T_qE_{ii})=BT^{-1}_qE_{\pi_n(i),\pi_n(i)}
\end{equation*}

\begin{equation}
\sigma(zE_{ii})=zAT^r_qE_{\pi_n(i),\pi_n(i)}
\end{equation}

\begin{equation*}
\sigma(z^{-1}E_{ii})=A^{-1}q^{r}z^{-1}T^{-r}_qE_{\pi_n(i),\pi_n(i)}
\end{equation*}

\begin{equation*}
\sigma(E_{ij})= \begin{cases} c_{i,j}E_{\pi_n(j),\pi_n(i)} & \text{if} \quad i>j \\
								c^{-1}_{j,i}E_{\pi_n(j),\pi_n(i)} & \text{if} \quad i<j
								\end{cases}
\end{equation*}
where $A$, $B$,  $c_{i,j}, \, r \, \in \, \mathbb{C}$, $A^2(Bq^{-1})^r=1$ and $c_{i,j}$ verify relations \eqref{eq:productosc} and \eqref{eq:lamolesta}.

\end{corollary}

\begin{proof}
If $n=N$ there is only case a to be considered in the proof of Theorem \ref{Teoremita}.
\end{proof}

\begin{remark}
Case $N=1$ coincides with \cite{BL05}.
\end{remark}


We will now concentrate on the implications of conditions \eqref{eq:productosc} and \eqref{eq:lamolesta}. First, let us note that as a consequence of \eqref{eq:productosc}, all coefficients $c_{i,j}$ are completely determined by

\begin{equation}
c_i := c_{i+1,i}, \quad i=1, \cdots, N-1
\end{equation}
and the upper condition of \eqref{eq:lamolesta} can be written as $c_i.c_{\pi_n(i+1)}=1 \, (i \neq n-1)$ by \eqref{eq:pichi}. Combining the lower condition of \eqref{eq:lamolesta} with \eqref{eq:productosc} we get $\pm 1=c_n.(c_{N,1})^{-1}=c_n.\prod_i (c_i)^{-1}= \prod_{i \neq n} (c_i)^{-1}$. Also, let us note that the permutation $\pi_n$ is given by two simple permutations of the sets $\{ 1, \cdots , n \}$ and $\{ n+ 1, \cdots , N \}$. Thus, Eq. \eqref{eq:lamolesta} reduces to

\begin{equation}\label{eq:condition}
c_i c_{n-i}=1 \, (1 \leq i \leq n), \quad c_{n+i}c_{N-i}=1 \, (1 \leq i < N-n)
\end{equation}
and
\begin{equation}\label{eq:lastcondition}
\pm 1=\prod_{i \neq n} c_i.
\end{equation}

Let $N=t+n$ and let us analyze the previous formulas. If $n$ (respectively, $t$) is even, by \eqref{eq:condition} we have $\prod_{i<n} c_i = c_{n/2}$ and $(c_{n/2})^2=1$ (respectively, $\prod_{i>n} c_i = c_{n+(t/2)}$ and $(c_{n+(t/2)})^2=1$). The coefficient $c_{n/2}$ (respectively, $c_{n+(t/2)}$) will be called a fixed point. 

\textit{Case -:}

If $N$ is even and
\begin{enumerate}
\item $n$ is even, the condition \eqref{eq:lastcondition} is satisfied if there are two fixed points: One of them must be $1$ and the other one equal to $-1$.
\item $n$ is odd, then there are no fixed points and \eqref{eq:lastcondition} is impossible. Thus \textit{there is no anti-involution in this case.}
\end{enumerate}  

If $N$ is odd, then $n$ or $t$ is even and we have only one fixed point that must be equal to $-1$.

\textit{Case+:}
For any $N$, condition \eqref{eq:lastcondition} is satisfied if the (possible) fixed points are all equal to $1$.

$ $\newline


From now on, we will consider separately the cases $N=n$ and $n < N$ in an attempt to exhibit more clearly their particular results.

\section{Case $n=N$.}

\subsection{Lie subalgebras of $\sqn$}

Let $\sqn^{A,B,c,r,N}$ denote the Lie subalgebra of $\sqn$ fixed by minus $\sigma_{A,B,c,r,N}$, namely

\begin{equation}
\sqn^{A,B,c,r,N}=\{ a \, \epsilon \, \sqn | \sigma_{A,B,c,r,N}(a)=-a \},
\end{equation}

where $\sigma_{A,B,c,r,N}$, for $h \, \in \, \mathbb{C}[w, w^{-1}]$, is given by

\begin{equation}\label{eq:sigmapuntito}
\sigma_{A,B,c,r,N}(z^kh(T_q)E_{i,j})=A^{k}q^{k(k-1)r/2}z^k h(Bq^{-k}T^{-1}_q)T^{kr}_qE_{\pi(j),\pi(i)}.
\end{equation}

Note that $\dot{\sigma}_{A,B,r}$ from \cite{BL05} agrees with $\sigma_{A,B,c,r,N}$ for $N=1$.

$ $\newline
Let us now analyze the relation among $\sqn^{A,B,c,r,N}$ for different values of $A,B,c,r$ and $N$. To that end, let $s \, \in \, \mathbb{C}$, denote by $\theta_s$ the automorphism of $\sqan$ given by $\theta_s(M)=M$, $\theta_s(zI)=zI$ and $\theta_s(T_qI)=q^sT_qI$, where $M \, \in \, Mat_N \mathbb{C} $ and $I$ stands for the identity matrix. It is easy to check that $\theta_s$ preserves the principal $\mathbb{Z}-$gradation of $\sqan$. Making use of the equation for $\sigma_{A,B,c,r,1}$ pointed out in \eqref{eq:sigmapuntito}, we have

\begin{equation}\label{eq:sigmatheta1}
\theta_{s}\sigma_{A,B,c,r,N}\theta_{-s}=\sigma_{q^{sr}A,q^{-2s}B,c,r,N}
\end{equation}

which resembles \cite{BL05}, when $N=1$.

Similarly, let $\alpha=\{\alpha_{i,j} \} \, (i>j)$ satisfying \eqref{eq:productosc} and \eqref{eq:lamolesta}. Denote by $\Gamma_{\alpha}$ the automorphism of $\sqan$ defined by $\Gamma_{\alpha}(zI)=zI$, $\Gamma_{\alpha}(T_qI)=T_qI$,

\begin{equation}
\Gamma_{\alpha}(E_{i,j})= \begin{cases} \alpha_{i,j}E_{i,j} & \text{if} \quad i>j \\
								 \alpha^{-1}_{j,i}E_{i,j} & \text{if} \quad i<j
								 \end{cases}
\end{equation}
Let $\sigma_c:=\sigma_{A,B,c,r,N}$, then we have

\begin{equation}\label{eq:sigmagamma1}
\sigma_c.\Gamma_{\alpha}=\sigma_{c.\alpha}=\Gamma_{\alpha^{-1}}.\sigma_c,
\end{equation}

$ $\newline
where $(c.\alpha)_{i,j}:=c_{i,j}\alpha_{i,j}$ and $(\alpha^{-1})_{i,j}=\alpha^{-1}_{i,j}$. Observe that $c.\alpha$ and $\alpha^{-1}$ also satisfy \eqref{eq:productosc} and \eqref{eq:lamolesta}.
Using \eqref{eq:sigmatheta1} and \eqref{eq:sigmagamma1}, we have:

\begin{lemma}
The Lie algebras $\sqn^{A,B,c,r,N}$ for arbitrary choices of $  A, \, B \text{ and } c$
are isomorphic to $\sqn^{\epsilon, q, \mathbf{1}, r, N}$, where $\epsilon$ is $1$ or $-1$, and $\mathbf{1}$ is the matrix $c$ with $c_i=1$ except for the fixed points that are $1$ or $-1$, which keep their sign.
\end{lemma}

We shall introduce some notation in  order to give an explicit description of this family of subalgebras.

First, we will write $\sigma_{\epsilon,r,N}$ and $\sqn^{\epsilon,r,N}$ instead of $\sigma_{\epsilon,q,\mathbf{1},r,N}$ and $\sqn^{A,q,\mathbf{1},r,N}$. Also, for any matrix $M \, \in \, Mat_{m \times n}(\mathbb{C})$, define

\begin{equation}
(M)_{i,j}^{\dagger}=M_{n+1-j,m+1-i},
\end{equation}

i.e., the transpose with respect to the ``other" diagonal. Recall the anti-involutions on $\sq:=\mathscr{S}_{q,1}$  given in \cite{BL05}:

\begin{equation}\label{eq:nomatricialN}
\dot{\sigma}_{\pm,B, r}(z^kf(T_q))=(\pm z)^kq^{k(k-1)r/2}f(Bq^{-k}T^{-1}_q)T_q^{kr}.
\end{equation}

An extension of $\dot{\sigma}_{\pm,B,r}$ to a map on $Mat_{N \times N}(\sq)=\sq \otimes Mat_{N \times N}(\mathbb{C})$ can be made by taking $[\dot{\sigma}_{\pm,B,r}(M)]_{i,j}=\dot{\sigma}_{\pm,B,r}(M_{i,j})$.

Case +:
We define the following map on $Mat_{N \times N}(\sq)$:

\begin{eqnarray}\label{eq:desdeaca}
M^{\dagger_1}=\dot{\sigma}_{+,q,r}(M^{\dagger})
\end{eqnarray}

Explicitly, the anti-involution $\sigma_{+,r,N}$ on $\sqn=\sq \otimes Mat_N (\mathbb{C})$ is given by

\begin{equation}
\sigma_{+,r,N} (M)= \left( M^{\dagger_1} \right),
\end{equation}

where $M \, \in \, Mat_{N \times N}(\sq)$, and

\begin{equation}\label{eq:hastaaca}
\sqn^{+,r,N}= \{ M  : M+M^{\dagger_1}=0 \}.
\end{equation}

Case -:
Now, consider the following map on $Mat_{N \times N}(\sq)$:

\begin{eqnarray}\label{eq:estetambien}
M^{\ast_1} :=\dot{\sigma}_{-,q,r}(M^{\dagger}).
\end{eqnarray}

Then $\sigma_{-,r,N}$ on $\sqn$ is explicitly given by

\begin{equation}
\sigma_{-,r,N} ( M)= \left( M^{\ast_1} \right),
\end{equation}

where $M \, \in \, Mat_{N \times N}(\sq)$. And

\begin{equation}\label{eq:hastaacaa}
\sqn^{-,r,N}=\{ M : M+M^{\ast_1}=0  \}.
\end{equation}

Let us note that $\sqn^{\pm,r,N}$ are Lie subalgebras of $\sqn$ and that $\dagger_1$ and $\ast_1$ are anti-automorphisms.

\begin{remark}\label{Remark} Replacing $\dagger$ by $T$ (usual transpose) in \eqref{eq:desdeaca} and \eqref{eq:estetambien} gives us another family of involutions that we shall denote by $\sigma^T_{\pm,r,N}$, which do not preserve the principal $\mathbb{Z}$-gradation. Moreover, the corresponding subalgebras are not $\mathbb{Z}$-graded subalgebras of $\sqn$, even though they are isomorphic to the others using that $M^{\dagger}=J M^TJ^{-1}$, where $J$ is the following $N \times N$ matrix

\begin{equation}\label{eq:ladejotaN}
J= \Bigg( \begin{matrix}
 0 & \cdots & 1 \\
\vdots & 1 & \vdots \\
1 & \cdots & 0
\end{matrix} \Bigg) .
\end{equation}

This way, we get $Ad_{J} \circ \sigma^T_{\pm,r,N} = \sigma_{\pm,r,N}$.
\end{remark}


\subsection{Generators of $\sqn^{\epsilon, r, N}$}

We can now give a detailed description of the generators of $\sqn^{\epsilon, r, N}$. 

Let us denote $ \mathbb{C} [ w,w^{-1} ] ^{(\epsilon),j}$ (where $\epsilon = 1$ or $\epsilon = -1$) the set of Laurent polynomials such that $f(w^{-1})=-(\epsilon)^jf(w)$.

Note that $\sqn^{\epsilon, r, N}= \{ x- \sigma_{\epsilon,r,N}(x):x \, \in \, \sqn \}$ and observe that by \eqref{eq:nomatricialN}

\begin{equation*}
\dot{\sigma}_{\pm,q,r}(z^k(q^{(k-1)/2}T_q)^{kr/2}f(q^{(k-1)/2}T_q))=(\pm z)^k (q^{(k-1)/2}T_q)^{kr/2} f((q^{(k-1)/2}T_q)^{-1}).
\end{equation*}

Here and in the following we will use the description of the elements in the subalgebras used in \eqref{eq:hastaaca} and \eqref{eq:hastaacaa}. The following is a set of generators of $\sqn^{\pm, r, N}$.

\begin{align*}
 \{ z^k (q^{(k-1)/2}T_q)^{kr/2}(f(q^{(k-1)/2}T_q)E_{i,n+1-j} -(\epsilon)^k & f((q^{(k-1)/2}T_q)^{-1})E_{j,n+1-i}): k \, \in \, \mathbb{Z}, \\
					&   \, f \, \in \mathbb{C}[w,w^{-1}], \, 1 \leq i < j \leq n \},
\end{align*}

and the generators on the opposite diagonal are

\begin{equation*}
\{ z^k (q^{(k-1)/2}T_q)^{kr/2} f(q^{(k-1)/2}T_q)E_{i,n+1-i}: k \, \in \, \mathbb{Z}, \, f \, \in \mathbb{C}[w,w^{-1}]^{(\epsilon),k}, \, 1 \leq i \leq n \}.
\end{equation*}

\subsection{Geometric realization of $\sigma_{\pm,r,N}$}

In this subsection we give a geometric realization of $\sigma_{\pm,r,N}$.
The algebra $\sqn$ acts on the space $V=\mathbb{C}^N[z, z^{-1}]$ and we define two bilinear forms on $V$:

\begin{equation}
B_{\pm}(h,g)=Res_z(\Phi_{\pm}(h^{T})Jg),
\end{equation}

where $J=z^{-2}J_N$, $J_N$ as in \eqref{eq:ladejotaN} and $\Phi_{\pm}: V \rightarrow V $ given by $\Phi_{\pm}(h(z))=h(\pm z), \, h(z) \, \in \, V$.


\begin{proposition}
\begin{itemize}
\item[(a)] The bilinear forms $B_{\pm}$ are nondegenerate. Moreover, $B_{+}$ is symmetric and $B_{-}$ is anti-simmetric.
\item[(b)] For any $L \, \in \, \sqn$ and $h,g \, \in \, V$ we have
\begin{equation}
B_{\pm}(Lh,g) = B_{\pm}(h, (T^{-kr/2}_q\sigma_{\pm,r,N}(L)T^{kr/2}_q)(g)),
\end{equation}
where $L=z^kT^{kr/2}_q p(T_q)(M)$. In other words, $L$ and $T^{-kr/2}_q\sigma_{\pm,r,N}(L)T^{kr/2}_q$ are adjoint operators with respect to $B_{\pm}$.
\end{itemize}
\end{proposition}

\begin{proof}
\begin{itemize}

\item[(a)] The statements are straightforward.

\item[(b)] Let $L=z^kT^{kr/2}_q p(T_q)(M)$, $h=z^ue_p$ and $g=z^se_q$. Recall that

\begin{equation*}
L(h)=z^{k+u}q^{rku/2}p(q^u)(Me_p)
\end{equation*}

and

\begin{equation*}
\sigma_{\pm,r,N}(L)(g)=(\pm1)^k z^{s+k} q^{skr/2} p(q^{-k-s+1})M^{\dagger} e_q.
\end{equation*}

So,

\begin{align}\label{eq:LHS}
B_{\pm}(L(z^ue_p),z^se_q)
&=Res_z (\pm z)^{k+u} q^{rku/2} p(q^u) e^T_p M^T z^{-2}J_N z^s e_q \\ \nonumber
&=(\pm 1)^{k+u} q^{rku/2} p(q^u) \delta_{k+u+s,1}(M^{T}J_N)_{(p,q)}. \nonumber
\end{align}

On the other hand, we have

\begin{align}\label{eq:RHS}
B_{\pm}(h, \sigma_{\pm,r,N}(L)g)
&=Res_z (\pm1)^{k+u} z^{u+k+s-2} q^{skr/2}p(q^{-k-s+1}) e^T_p J_N M^{\dagger} e_q \\ \nonumber
&=((\pm1)^{k+u} q^{skr/2} \delta_{k+u+s,1}p(q^{-k-s+1})J_NM^{\dagger} )_{(p,q)} \\ \nonumber
&=((\pm1)^{k+u} q^{skr/2} p(q^{-k-s+1}) \delta_{k+u+s,1}J_NM^{\dagger} )_{(p,q)} \nonumber
\end{align}
\end{itemize}

Note that if we multiply \eqref{eq:LHS} by $q^{skr/2}$ and \eqref{eq:RHS} by $q^{rku/2}$, we get

\begin{equation}\label{eq:fbconq}
B_{\pm}(L(z^ue_p),T^{kr/2}_q I z^se_q)=B_{\pm}(T^{kr/2}_q I z^ue_p, \sigma_{\pm,r,N}(L) z^se_q).
\end{equation}

It is easy to prove that, for $\alpha \, \in \, \mathbb{C}$,

\begin{equation*}
B_{\pm}(T^{\alpha}_qIz^ue_p, z^se_q)=B_{\pm}(z^ue_p, \sigma_{\pm,r,N}(T^{\alpha}_qI) z^se_q).
\end{equation*}

Making use of this result in \eqref{eq:fbconq}, we can see that

\begin{equation*}
B_{\pm}(\sigma_{\pm,r,N}(T^{kr/2}_qI)L(z^ue_p), z^se_q)=B_{\pm}(z^ue_p, \sigma_{\pm,r,N}(T^{kr/2}_qI)\sigma_{\pm,r,N}(L) z^se_q).
\end{equation*}

Thus, as expected, we get

\begin{equation*}
B_{\pm}(L(z^ue_p), z^se_q)=B_{\pm}(z^ue_p, (T^{-kr/2}_qI)\sigma_{\pm,r,N}(L)(T^{kr/2}_qI) z^se_q).
\end{equation*}

\end{proof}

\begin{remark}
In a similar fashion, we can define the following nondegenerate bilinear forms on $V$:
\begin{equation*}
B^T_{\pm}(h,g)=Res_z (\Phi_{\pm}(h^T) J_{T} g),
\end{equation*}

where

\begin{equation*}
J_{T}= z^{-2}I_n,
\end{equation*}

with $I_n$ the $n \times n$ identity matrix, and it easily follows that they satisfy

\begin{equation*}
B_{\pm}(Lh,g)=B_{\pm}(h,T^{-kr/2}_q\sigma^T_{\pm,r,N}(L)T^{kr/2}_qg),
\end{equation*}

where $\sigma^T_{\pm,n}$ were defined in Remark \ref{Remark}.

\end{remark}


$ $\newline
\section{Case $n < N$.}
\subsection{Lie subalgebras of $\sqn$}

Let $\sqn^{\pm,B,c,n}$ denote the Lie subalgebra of $\sqn$ fixed by minus $\sigma_{\pm,B,c,n}$:

\begin{equation}
\sqn^{\pm,B,c,n}=\{ a \, \epsilon \, \sqn | \sigma_{\pm,B,c,n}(a)=-a \}.
\end{equation}

$ $\newline
As in the case $n=N$, we analyze the relation among $\sqn^{\pm,B,c,n}$ for different values of $B,c$ and $n$. Let $s \, \epsilon \, \mathbb{C}$, denote by $\theta_s$ the automorphism of $\sqan$ given by $\theta_s(M)=M$, $\theta_s(zI)=zI$ and $\theta_s(T_qI)=q^sT_qI$, where $I$ stands for the identity matrix and $M \, \epsilon \, Mat_N \mathbb{C} $. Clearly $\theta_s$ preserves the principal $\mathbb{Z}-$gradation of $\sqan$. As before, we have for this case the following:

\begin{equation}\label{eq:sigmatheta}
\theta_{s}\sigma_{\pm,B,c,n}\theta_{-s}=\sigma_{\pm,q^{-2s}B,c,n}.
\end{equation}

Let $\alpha=\{\alpha_{i,j} \} \, (i>j)$ satisfying \eqref{eq:productosc} and \eqref{eq:lamolesta} and denote by $\Gamma_{\alpha}$ the automorphism of $\sqan$ defined by $\Gamma_{\alpha}(zI)=zI$, $\Gamma_{\alpha}(T_qI)=T_qI$,

\begin{equation}
\Gamma_{\alpha}(E_{i,j})= \begin{cases} \alpha_{i,j}E_{i,j} & \text{if} \quad i>j \\
								 \alpha^{-1}_{j,i}E_{i,j} & \text{if} \quad i<j
								 \end{cases}
\end{equation}
Letting $\sigma_c:=\sigma_{\pm,B,c,n}$, we have

\begin{equation}\label{eq:sigmagamma}
\sigma_c.\Gamma_{\alpha}=\sigma_{c.\alpha}=\Gamma_{\alpha^{-1}}.\sigma_c,
\end{equation}
where $(c.\alpha)_{i,j}:=c_{i,j}\alpha_{i,j}$ and $(\alpha^{-1})_{i,j}=\alpha^{-1}_{i,j}$. Note that $c.\alpha$ and $\alpha^{-1}$ also satisfy \eqref{eq:productosc} and \eqref{eq:lamolesta}.
Making use of \eqref{eq:sigmatheta} and \eqref{eq:sigmagamma}, we have:

\begin{lemma}
The Lie algebras $\sqn^{\pm,B,c,n}$ for arbitrary choices of $  \, B \text{ and } c$
are isomorphic to $\sqn^{\epsilon, q, \mathbf{1}, n}$, where $\epsilon$ is $1$ or $-1$, and $\mathbf{1}$ is the matrix $c$ with $c_i=1$ except for the fixed points that are $1$ or $-1$, which keep their sign.
\end{lemma}

\begin{remark}
Due to this Lemma, we can find a complex number $s$ such that $B=q$. Moreover, recall that in the case $n = N$, we find a complex number such that $A = \pm 1$ (which  makes $B = q$ as a consequence). So, in both cases, the subalgebra $\sqn^{A,B,c,r,n}$ is isomorphic to $\sqn^{\pm 1, q, \mathbf{1}, r, n}$ and the only distinction between both cases is regarding $r$: while $r$ takes an arbitrary value in the case $n=N$, if $n<N$ $r$ happens to be $0$.
\end{remark}

We will write $\sigma_{\pm,n}$ and $\sqn^{\pm,n}$ instead of $\sigma_{\pm,B,c,n}$ and $\sqn^{\pm,B,c,n}$, with $B=q$ and $c=\mathbf{1}$. As in the previous section, for any matrix $M \, \in \, Mat_{m \times n}(\mathbb{C})$, we define

\begin{equation}
(M)_{i,j}^{\dagger}=M_{n+1-j,m+1-i},
\end{equation}

i.e., the transpose with respect to the ``other" diagonal. Recall once again the anti-involutions on $\sq:=\mathscr{S}_{q,1}$  given in \cite{BL05}:

\begin{equation}\label{eq:nomatricial}
\dot{\sigma}_{\pm,B,r}(z^kf(T_q))=(\pm z)^kq^{k(k-1)r/2}f(Bq^{-k}T^{-1}_q)T_q^{kr}.
\end{equation}

We extend $\dot{\sigma}_{A,B,r}$ to a map on $Mat_{a \times b}(\sq)=\sq \otimes Mat_{a \times b}(\mathbb{C})$ by taking $[\dot{\sigma}_{\pm,B,r}(M)]_{i,j}=\dot{\sigma}_{\pm,B,r}(M_{i,j})$. Now let $t=N-n$.

Case +:
We define the following maps:

\begin{eqnarray}\label{eq:desdeaca2}
M^{\dagger_1}=\dot{\sigma}_{+,q,0}(M^{\dagger}), \quad B^{\dagger_2}=z^{-1}\dot{\sigma}_{+,q,0}(B^{\dagger}) \\ \nonumber
C^{\dagger_3}=z\dot{\sigma}_{+,1,0}(C^{\dagger}), \quad  D^{\dagger_4}=\dot{\sigma}_{+,1,0}(D^{\dagger}),
\end{eqnarray}
where $M \, \in \, Mat_{n \times n}(\sq)$, $B \, \in \, Mat_{n \times t}(\sq)$, $C \, \in \, Mat_{t \times n}(\sq)$, and $D \, \in \, Mat_{t \times t}(\sq)$.
We can write the anti-involution $\sigma_{+,n}$ on $\sqn=\sq \otimes Mat_N (\mathbb{C})$ explicity as

\begin{equation}
\sigma_{+,n} \left( \begin{matrix}
 						M & B \\
						C & D
						\end{matrix} \right)= \left( \begin{matrix}
 														M^{\dagger_1} & C^{\dagger_3} \\
														B^{\dagger_2} & D^{\dagger_4}
														\end{matrix} \right),
\end{equation}

and

\begin{equation}\label{eq:hastaaca2}
\sqn^{+,n}= \bigg\{ \left( \begin{matrix}
 						M & B \\
						-B^{\dagger_2} & D
						\end{matrix} \right) : M+M^{\dagger_1}=0  \, \text{ and } D+D^{\dagger_4}=0 \bigg\}.
\end{equation}


The fact that $\sigma_{+,n}(a)=-a$ implies $C^{\dagger_3}=-B$ and $B^{\dagger_2}=-C$, and these two conditions are equivalent because $(B^{\dagger_2})^{\dagger_3}=B$. Moreover, proving that $\sqn^{+,n}$ is a Lie subalgebra of $\sqn$ by direct computations requires using that $^{\dagger_1}$ and $^{\dagger_4}$ are anti-automorphisms, and the identities $B^{\dagger_2}=z^{-1}B^{\dagger_1}$, $C^{\dagger_4}=zC^{\dagger_3}$, $(B^{\dagger_2})^{\dagger_1}=Bz^{-1}$, $B^{\dagger_4}z^{-1}=B^{\dagger_2}$, etc. Observe, however, that $^{\dagger_2}$ and $^{\dagger_3}$ are not anti-automorphisms. The following identities are also useful:

\begin{eqnarray}
\dot{\sigma}_{+,q,0}(z^{-1}\dot{\sigma}_{+,q,0}(z^kf(T_q)))=(z^kf(T_q))z^{-1},\\ \nonumber
\dot{\sigma}_{+,1,0}(z^{-1}\dot{\sigma}_{+,q,0}(z^kf(T_q)))= z^{-1}(z^kf(T_q)).
\end{eqnarray}

$ $\newline

Case -:
As we have seen in the analysis following equation \eqref{eq:lastcondition}, the case $N$ even and $n$ (also $t$) odd is impossible. Therefore we may suppose, due to the symmetry, that $t$ is even. Now, we shall consider the following maps:

\begin{eqnarray}\label{eq:estetambien2}
M^{\ast 1} :=\dot{\sigma}_{-,q,0}(M^{\dagger}), \\ \nonumber
B^{\ast 2}:=z^{-1}\dot{\sigma}_{-,q,0}({B^{\dagger}} ), \\ \nonumber
C^{\ast 3}:=z\dot{\sigma}_{-,1,0}(C^{\dagger}), \\
D^{\ast 4}:=\dot{\sigma}_{-,1,0} (D^{\dagger}), \nonumber
\end{eqnarray}
where $M \, \in \, Mat_{n \times n}(\sq)$, $B \, \in \, Mat_{n \times t}(\sq)$, $C \, \in \, Mat_{t \times n}(\sq)$, and $D \, \in \, Mat_{t \times t}(\sq)$.
Then the anti-involution $\sigma_{-,n}$ on $\sqn$ is explicitly given by

\begin{equation}
\sigma_{-,n} \left( \begin{matrix}
 						M & B \\
						C & D
						\end{matrix} \right)= \left( \begin{matrix}
 														M^{\ast 1} & C^{\ast 3} \\
														B^{\ast 2} & D^{\ast 4}
														\end{matrix} \right),
\end{equation}

and

\begin{equation}\label{eq:hastaacaa2}
\sqn^{-,n}= \bigg\{ \left( \begin{matrix}
 						M & B \\
						-B^{\ast 2} & D
						\end{matrix} \right) : M+M^{\ast 1}=0  \, \text{ and } D+D^{\ast 4}=0 \bigg\}.
\end{equation}

As before, condition $\sigma_{-,n}(a)=-a$ implies $C^{\ast 3}=-B$ and $B^{\ast 2}=-C$, and these two conditions are equivalent due to the fact that $(B^{\ast 2})^{\ast 3}=B$. Moreover, to prove that $\sqn^{-,n}$ is a Lie subalgebra of $\sqn$ by direct computations, requires using that $^{\ast 1}$ and $^{\ast 4}$ are anti-automorphisms, and the identities $B^{\ast 2}=z^{-1}B^{\ast 1}$, $D^{\ast 4}=zD^{\ast 3}$, $(B^{\ast 2})^{\ast 1}=Bz^{-1}$, $B^{\ast 4}z^{-1}=B^{\ast 2}$, etc. Once again, $^{\ast 2}$ and $^{\ast 3}$ are not anti-automorphisms. We also need to use:

\begin{eqnarray}
\dot{\sigma}_{-,q,0}(z^{-1}\dot{\sigma}_{-,q,0}(z^kf(T_q)))=-(z^kf(T_q))z^{-1},\\ \nonumber
\dot{\sigma}_{-,1,0}(z^{-1}\dot{\sigma}_{-,q,0}(z^kf(T_q)))=- z^{-1}(z^kf(T_q)).
\end{eqnarray}
$ $\newline

\begin{remark} Replacing $\dagger$ by $T$ (usual transpose) in \eqref{eq:desdeaca2} and \eqref{eq:estetambien2} gives another family of involutions denoted by $\sigma^T_{\pm,n}$. These involutions do not preserve the principal $\mathbb{Z}$-gradation, and the corresponding subalgebras are not $\mathbb{Z}$-graded subalgebras of $\sqn$, but they are isomorphic to the others using the same argument in Remark \ref{Remark}.
\end{remark}


\subsection{Generators of $\sqn^{\pm, n}$}

In this subsection we give a detailed description of the generators of $\sqn^{\pm, n}$. 

Let us denote $ \mathbb{C} [ w,w^{-1} ] ^{(\epsilon),j}$ (where $\epsilon = 1$ or $\epsilon = -1$ the set of Laurent polynomials such that $f(w^{-1})=-(\epsilon)^jf(w).$ And let $\bar{l}=0$ if $l$ is odd and $\bar{l}=1$ if $l$ is even.

Recall that

\begin{equation*}
\dot{\sigma}_{\pm,q, 0}(z^kf(q^{(k-1)/2}T_q))=(\pm z)^kf((q^{(k-1)/2}T_q)^{-1})
\end{equation*}

and also,

\begin{equation*}
\dot{\sigma}_{\pm,1,0}(z^kf(q^{k/2}T_{q}))=(\pm z)^k f(q^{-k/2}T^{-1}_{q}).
\end{equation*}

Therefore, the following is a set of generators of $\sqn^{\pm, n}$, using the description of the elements in the subalgebras given in \eqref{eq:hastaaca2} and \eqref{eq:hastaacaa2}.

\begin{itemize}
\item For block M, where $1 \leq i, \, j \leq n$:
\begin{equation*}
\{ z^k (f(q^{(k-1)/2}T_q)E_{i,n+1-j}-(\epsilon)^k f((q^{(k-1)/2}T_q)^{-1})E_{j,n+1-i}):  k \, \in \, \mathbb{Z}, \, f \, \in \mathbb{C}[w,w^{-1}], \, 1 \leq i < j \leq n \}.
\end{equation*}

and the generators on the opposite diagonal are
\begin{equation*}
\{ z^k f(q^{(k-1)/2}T_q)E_{i,n+1-i}: k \, \in \, \mathbb{Z}, \, f \, \in \mathbb{C}[w,w^{-1}]^{(\epsilon),k}, \, 1 \leq i \leq n \}.
\end{equation*}

\item For blocks $B$ (and $C$), where $i \leq n$ and $j>n$ (or $j \leq n$ and $i>n$):
\begin{align*}
\{ z^k (f(q^{(k-1)/2}T_q)  &  E_{i,n+j}-(\epsilon)^kz^{-1}f((q^{(k-1)/2}T_q)^{-1})E_{N+1-j,n+1-i}):  k \, \in \, \mathbb{Z}, \, f \, \in \mathbb{C}[w,w^{-1}], \\
							&  \, 1 \leq i \leq n \, \text{ and } 1 \leq j \leq N-n \}.
\end{align*}

\item For block $D$, where $i,j>n$:
\begin{align*}
\{ z^k (f(q^{k/2}T_{q}) & E_{n+i,N+1-j}-f(q^{-k/2}T^{-1}_{q})E_{n+j,N+1-i}):  k \, \in \, \mathbb{Z}, \, f \, \in \mathbb{C}[w,w^{-1}], \, 1 \leq i < j \leq N-n \},
\end{align*}

and the generators on the opposite diagonal are
\begin{align*}
\{ z^k f(q^{k/2}T_{q})E_{n+i,N+1-i}:  k \, \in \, \mathbb{Z}, \, f \, \in \mathbb{C}[w,w^{-1}]^{(\epsilon),k}, \, 1 \leq i \leq N-n \}.
\end{align*}

\end{itemize}


\subsection{Geometric realization of $\sigma_{\pm,n}$}

In this section we give a geometric realization of $\sigma_{\pm,n}$.

The algebra $\sqn$ acts on the space $V=\mathbb{C}^N[z, z^{-1}]$ and we define two bilinear forms on $V$:

\begin{equation}
B_{\pm}(h,g)=Res_z(\Phi_{\pm}(h^{T})Jg),
\end{equation}

where

\begin{equation*}
J=\Bigg( \begin{matrix}
 z^{-2}J_n & 0  \\
 0 & z^{-1}J_t
\end{matrix} \Bigg),
\end{equation*}
with $\Phi: V \rightarrow V $ given by $\Phi_{\pm}(h(z))=h(\pm z), \, h(z) \, \in \, V$, and $J_n$ as in \eqref{eq:ladejotaN}. Observe that $V=\mathbb{C}^n[z, z^{-1}] \times \mathbb{C}^t[z, z^{-1}]$ is an orthogonal decomposition of $V$. Now, consider the following proposition.


\begin{proposition}
\begin{itemize}
\item[(a)] The bilinear forms $B_{\pm}$ are nondegenerate. Moreover, $B_{+}$ is symmetric and $B_{-}$ is symmetric in the subspace $\mathbb{C}^n[z, z^{-1}]$ and anti-symmetric in $\mathbb{C}^t[z, z^{-1}]$.
\item[(b)] For any $L \, \in \, \sqn$ and $h,g \, \in \, V$ we have
\begin{equation}
B_{\pm}(Lh,g) = B_{\pm}(h, \sigma_{\pm,n}(L)g),
\end{equation}
that is, $L$ and $\sigma_{\pm,n}(L)$ are adjoint operators with respect to $B_{\pm}$.
\end{itemize}
\end{proposition}

\begin{proof}
\begin{itemize}

\item[(a)] The statements are straightforward.

\item[(b)] Let $L=z^kp(T_q)(^{M B}_{C D})$, $h=z^ue_p$, and $g=z^se_q$ be as shown previously. Recall that

\begin{equation*}
L(h)=z^{k+u}p(q^u)(^{M B}_{C D})e_p
\end{equation*}

and

\begin{equation*}
\sigma_{\pm,n}(L)(g)=(\pm1)^k z^{s+k} \Big( \begin{matrix}
p(q^{-k-s+1})M^{\dagger} & z p(q^{-k-s})C^{\dagger} \\
z^{-1}p(q^{-k-s+1})B^{\dagger} & p(q^{-k-s})D^{\dagger}
\end{matrix} \Big) e_q.
\end{equation*}

So,

\begin{align*}
B_{\pm}(L(z^ue_p),z^se_q)
&=Res_z (\pm1)^{z+u}z^{k+u} p(q^u) e^T_p \Big( \begin{matrix}
M & B \\
C & D
\end{matrix} \Big)^T \Big( \begin{matrix}
z^{-2}J_n & 0 \\
0 & z^{-1}J_t
\end{matrix} \Big) z^s e_q \\
&=(\pm1)^{z+u} p(q^u)\Big( \begin{matrix}
 \delta_{k+u+s,1}M^{T}J_n & \delta_{k+u+s,0}C^{T}J_t \\
 \delta_{k+u+s,1}B^{T}J_n & \delta_{k+u+s,0}D^{T}J_t
\end{matrix} \Big)_{(p,q)}.
\end{align*}

On the other hand, we have

\begin{align*}
B_{\pm}(h, \sigma_{\pm,n}(L)g)
&=Res_z (\pm1)^{z+u} z^{k+u+s} e^T_p \Big( \begin{matrix}
z^{-2}J_n & 0 \\
0 & z^{-1}J_t
\end{matrix} \Big) \Big( \begin{matrix}
p(q^{-k-s+1})M^{\dagger} & z p(q^{-k-s})C^{\dagger} \\
z^{-1}p(q^{-k-s+1})B^{\dagger} & p(q^{-k-s})D^{\dagger}
\end{matrix} \Big) e_q \\
&=(\pm1)^{z+u}\Big( \begin{matrix}
\delta_{k+u+s,1}p(q^{-k-s+1})J_nM^{\dagger} & \delta_{k+u+s,0} p(q^{-k-s})J_nC^{\dagger} \\
\delta_{k+u+s,1}p(q^{-k-s+1})J_tB^{\dagger} & \delta_{k+u+s,0}p(q^{-k-s})J_tD^{\dagger}
\end{matrix} \Big)_{(p,q)} \\
&=(\pm1)^{z+u}p(q^u) \Big( \begin{matrix}
\delta_{k+u+s,1}J_nM^{\dagger} & \delta_{k+u+s,0} J_nC^{\dagger} \\
\delta_{k+u+s,1}J_tB^{\dagger} & \delta_{k+u+s,0}J_tD^{\dagger}
\end{matrix} \Big)_{(p,q)}
\end{align*}
\end{itemize}

As the last two results are equal, we finish the proof.
\end{proof}

\begin{remark}
In a similar fashion, we can define the following nondegenerate bilinear forms on $V$:
\begin{equation*}
B^T_{\pm}(h,g)=Res_z (\Phi_{\pm}(h^T) J_{T} g),
\end{equation*}

where

\begin{equation*}
J_{T}=\Bigg( \begin{matrix}
 z^{-2}I_n & 0  \\
 0 & z^{-1}I_t
\end{matrix} \Bigg),
\end{equation*}

with $I_n$ the $n \times n$ identity matrix, and it easily follows that they satisfy

\begin{equation*}
B_{\pm}(Lh,g)=B_{\pm}(h,\sigma^T_{\pm,n}(L)g),
\end{equation*}

where $\sigma^T_{\pm,n}$ were defined in \eqref{eq:ladejotaN}.

\end{remark}



\end{document}